\documentclass[conference,a4paper]{IEEEtran}

\usepackage[utf8]{inputenc} 
\usepackage[T1]{fontenc}
\usepackage{url}
\usepackage{ifthen}
\usepackage{cite}

\usepackage{amsmath,amsthm,amssymb, cite}
\usepackage{bm,bbm}
\usepackage{mathtools}
\usepackage{multirow,color,amsfonts}
\usepackage{tabulary}
\usepackage{subfigure}
\usepackage{graphicx}
\usepackage{setspace}
\usepackage{enumerate}
\usepackage[]{algorithm2e}
\usepackage{comment}

\newtheorem{theo}{Theorem}
\newtheorem{defi}{Definition}

\newtheorem{prop}{Proposition}
\newtheorem{remark}{Remark}

\newtheorem{cor}{Corollary}

\DeclareMathOperator*{\etaM}{\eta_{\sf MatII}}
\DeclareMathOperator*{\etaGM}{\eta_{\sf SSCC}}

\DeclareMathOperator*{\Hpi}{H_{\pi}}
\DeclareMathOperator*{\Hpii}{H_{\pi,i}}
\DeclareMathOperator*{\HpiI}{H_{\pi,\mathcal{I}}}

\DeclareMathOperator*{\Rsp}{R_{\sf Sph}} 
\DeclareMathOperator*{\Rdd}{R_{\sf c}} 
 
\newcommand{\etapi}{\eta_{\pi}}

\newcommand{\Varpi}[1]{\mathrm{Var}_{\pi}\!\left[{#1} \right]}

\begin{document}
\title{Spatial Soft-Core Caching}

\author{
  \IEEEauthorblockN{Derya Malak and Muriel M\'{e}dard}
  \IEEEauthorblockA{Research Laboratory of Electronics, MIT, 
                    Cambridge, MA, USA\\
                    Email: \{deryam,\,medard\}@mit.edu}
  \and
  \IEEEauthorblockN{Edmund M. Yeh}
  \IEEEauthorblockA{Northeastern University, 
                    Boston, MA, USA\\
                    Email: eyeh@ece.neu.edu}
}

\maketitle

%%%%%%
\begin{abstract}
We propose a decentralized spatial soft-core cache placement (SSCC) policy for wireless networks. SSCC yields a spatially balanced sampling via negative dependence across caches, and can be tuned to satisfy cache size constraints with high probability. Given a desired cache hit probability, we compare the 95\% confidence intervals of the required cache sizes for independent placement, hard-core placement and SSCC  policies. We demonstrate that in terms of the required cache storage size, SSCC can provide up to more than 180\% and 100\% gains with respect to the independent and hard-core placement policies, respectively. SSCC can be used to enable proximity-based applications such as device-to-device communications and peer-to-peer networking as it promotes the item diversity and reciprocation among the nodes. 
\end{abstract}

%%%
\section{Introduction}
Distributed caching is a powerful technique to minimize the total average delay by replacing the backhaul capacity with storage capacity at small cells \cite{Shanmugam2013}, and to enable spectral reuse and throughput gain in networks \cite{MaddahAli2013Journal}. The goal of an efficient cache placement is to maximize the hit probability, i.e. the probability of obtaining the desired item from a neighboring cache. This is affected by the demand distribution, network topology, range of communication, and cache storage size.

Fundamental limits of caching have been studied in \cite{MaddahAli2013Journal}, in which the content placement phase is carefully designed so that a single coded multicast transmission can satisfy different demands. Capacity scaling laws have been explored in \cite{MaddahAli2013Journal}, and rate-memory and storage-latency tradeoffs have been studied in \cite{yu2018exact}. Caching has been studied in the context of device-to-device (D2D) communications in \cite{Ji2014}, and interference management in \cite{maddah2015cache}, and in optimization of cloud and edge processing for radio access networks in \cite{park2016joint}.

Temporal caching models have been analyzed in \cite{Che2002} for popular cache replacement algorithms, e.g. least recently used (LRU), least-frequently used, and most recently used cache update. Decentralized spatial LRU caching strategies have been developed in \cite{Giovanidis2016}. These combine the temporal and spatial aspects of caching, and approach the performance of centralized policies as the coverage increases. However, they are restricted to the LRU principle. A time-to-live (TTL) policy with a stochastic capacity constraint and low variance has been proposed in \cite{BerHenCiuSch2015}. The BitTorrent protocol employs the rarest first and choke algorithms to promote diversity of the pieces among peers, and foster reciprocation, respectively. These have been demonstrated in the context of peer-to-peer (P2P) file replication in the Internet \cite{LegKelMic2006}. A good piece replication algorithm should minimize the time spent in the transient state. 

There exist studies focusing on decentralized (geographic) content placement policies such as \cite{Shanmugam2013}, \cite{Blaszczyszyn2014}, \cite{IoannidisYeh2016}, \cite{Malak2016twc}. The main focus of the literature in this direction is to maximize the average cache hit probability subject to an average cache constraint. This optimization problem can be solved as a convex program. However, to the best of our knowledge, the related literature does not provide guarantees in terms of (a) how far-off the average cache size is from reality, (b) how far-off the average cache hit rate is from reality, and (c) how stable the cache hit probability across the caches. 

In the current paper, we provide a decentralized spatial soft-core cache placement (SSCC) policy. Since the cache storage size is finite, it is intuitive to have an exclusion range-based caching model such that the caches storing the same item are never closer to each other than some given distance (negative dependence), so as to promote diversity and reciprocation. SSCC roots in spatially balanced sampling, which is motivated by the request arrivals. For example, in P2P networking, the actual demand distribution is not known by nodes, and the cache updates in each peer are triggered by the requests. Furthermore, the traffic density in cellular networks is in general not uniform across the network, and the peak hour density can be approximated by a log-normal distribution \cite{zhou2015spatial}. Hence, instead of having a fixed exclusion range, it is desirable to have a variable exclusion range, depending on the popularity of the item. The SSCC policy come to the fore by putting a mark distribution on the exclusion range of an item based on its popularity. The marks may correspond to the detection ranges or the transmit powers of the nodes in heterogeneous network scenarios. Our objective is to address the issues (a)-(c) above in order to provide a better trade-off between the actual cache hit rate and the cache size violation probability. Our main contributions and use cases of SSCC are: 
\vspace{-0.1cm}
\begin{enumerate}[i.]
    \item SSCC has desirable properties: spatially balanced sampling across caches, concentration of the cache size, better cache over-provisioning, and multi-hop connectivity. 
    \item SSCC yields a better cache hit probability-cache violation probability tradeoff than the state of the art. In terms of the required cache storage size, SSCC can provide more than 180\% and 100\% gains with respect to independent placement \cite{Blaszczyszyn2014}, and hard-core placement \cite{Malak2016twc}, respectively.
    \item SSCC is suited for enabling proximity-based applications (D2D, P2P), and offloading mobile users in networks.  
    \item SSCC has connections with rarest first caching as it promotes the item diversity and reciprocation among the nodes. Hence, it can be well-suited for P2P applications.
\end{enumerate}

{\bf Notation.} 
Let $\Phi$ denote the mother point process (p.p.), and $\Phi_{th}$ be the child p.p. obtained via the thinning of $\Phi$. Let $\pi$ be a spatial caching policy that yields a set of child p.p.'s $\{\Phi_{th,i}\}_i$, where $\Phi_{th,i}$ is the set of retained points that cache item $i$. Let $A$ be a given bounded convex set in $\mathbb{R}^2$ containing the origin, and $rA$ be its dilation by the factor $r$. $\mathbbm{1}\{A\}$ is the indicator of event $A$. Let $B$ be a bounded Borel set. Let $\Phi(B)$ be the random number of points of the spatial p.p. $\Phi$ which lie in $B$. Any receiver can obtain the desired item $i$ if it is within a critical communication range $\Rdd$. Assume that $B=B_0(\Rdd)$, where $B_0(r)$ is a ball in $\mathbb{R}^2$ with radius $r$, centered at origin.

%%%
\section{How to Optimize the Caching Gain}
The locations of the nodes (caches) in the network are modeled by a homogeneous Poisson point process (PPP) $\Phi$ in $\mathbb{R}^2$ with intensity $\lambda$. There are $M$ items in the network, each having the same size, and each node has the same cache storage size $N<M$. Each user makes requests based on a Zipf popularity distribution over the set of the items. The probability mass function (pmf) of such requests (demand) is given by $p_r(i)=i^{-\gamma_r}\big/\sum\nolimits_{j=1}^M {j^{-\gamma_r}}$, where $\gamma_r$ determines the tilt of the Zipf distribution. The demand profile is the Independent Reference Model (IRM), i.e., the standard synthetic traffic model in which the request distribution does not change over time \cite{traverso2013temporal}. The request distribution is uniform across the network, i.e., isotropic, and does not change over time. Hence, the intensity of the requests for item $i$, i.e. $\lambda_i$, is proportional to its demand probability $p_r(i)$. Let $\mathcal{I}\sim p_r$ be the random variable that models the demand. Each node is associated with the variables $z_{xi}=\mathbbm{1}\{i\in {\rm Cache}(x)\}$ that denote whether item $i$ is available in its cache or not. There is also a cost $w_k$ associated with obtaining an item within the presence of $k$ nodes in the range. Given these parameters, consider the caching gain function of the following form: 
\begin{align}
\label{generalcachecost}
F(Z)=\mathbb{E}_{\mathcal{I}}\left[\sum\limits_{k=1}^{\infty}{w_{k}\Big(1-\prod\limits_{k'=1}^k\big(1-z_{p_{k'} I}\big)\Big)}\right],
\end{align}
where (\ref{generalcachecost}) can be used to model multi-hop coverage scenarios as in \cite{IoannidisYeh2016}, and Boolean Model coverage scenarios as in \cite{Blaszczyszyn2014}, \cite{Malak2016twc}. Let $\lambda_i=p_r(i)$, and $w_k=\mathbb{P}(\Phi(B)=k)$, which is the probability that $k$ caches (nodes of the original p.p. $\Phi$) cover the typical receiver, and $w_0$ is the probability of having no connection. Assume that $k^*$ is the first index such that a transmitter has the desired item $i$. Then, from (\ref{generalcachecost}), the caching gain for item $i$ is $\sum\nolimits_{k=k^*}^{\infty}{w_{k}}=\mathbb{P}(\Phi(B)\geq k^*)$, which is the same as probability of having at least $k^*$ transmitters. Equivalently, the cost of caching is $\sum\nolimits_{k=1}^{k^*-1}{w_{k}}$.

Since both the multi-hop and Boolean coverage scenarios are equivalent up to scaling, we focus on the second scenario. 

\begin{comment}
We have the following immediate observations: 
\begin{enumerate}[i.]
\item $F(Z)$ is an increasing function of $Z$.
\item The product term in (\ref{generalcachecost}) satisfies the following relation:
\begin{align}
\label{product_observation}
\!\!\!\!\!\!\!\!\!\!\prod\limits_{k'=1}^k\big(1-z_{p_{k'} i}\big)=\begin{cases}
1,\,\, z_{p_{k'} i}=0,\,\, \forall k'\in\{1,\hdots,k\},\\
0,\,\, \text{otherwise}.
\end{cases}
\end{align}
\item From (\ref{product_observation}), we get $\mathbb{E}\Big[\prod\limits_{k'=1}^k\big(1-z_{p_{k'} i}\big)\Big]
=\mathbb{P}(\Phi_{th,i}(B)=0)$.
\end{enumerate}
\end{comment}

We have the following immediate observation.
\begin{prop}\label{NA}
$F(Z)$ is convex if $z_{p_k' i}$'s are negatively associated (NA) \cite{Dubhashi1996p2} across $k'\in\{1,\hdots,k\}$, for all $i\in\{1,\hdots, M\}$. 
\end{prop}

\begin{proof}
Exploiting (\ref{generalcachecost}), we have the following relation: $\mathbb{E}[F(Z)]=\mathbb{E}_{\mathcal{I}}\left[\sum\nolimits_{k=0}^{\infty}w_k{\Big(1-\mathbb{E}\Big[\prod\nolimits_{k'=1}^k\big(1-z_{p_{k'} \mathcal{I}}\big)\Big]\Big)} \right]
\overset{(a)}{\geq} \mathbb{E}_{\mathcal{I}}\left[\sum\nolimits_{k=0}^{\infty}w_k{\Big(1-\prod\nolimits_{k'=1}^k\big(1-\mathbb{E}[z_{p_{k'} \mathcal{I}}]\big)\Big)} \right]=F(\mathbb{E}[Z])$, where $(a)$ is due to that $\mathbb{E}\Big[\prod\nolimits_{k'=1}^k\big(1-z_{p_{k'} i}\big)\Big]\leq\prod\nolimits_{k'=1}^k\big(1-\mathbb{E}[z_{p_{k'} i}]\big)$ as $z_{p_{k'} i}$'s are NA across $k'\in\{1,\hdots,k\}$, $\forall$ $i$. 
\end{proof}

From Prop. \ref{NA}, $\mathbb{E}[F(Z)]\geq F(\mathbb{E}[Z])$. The expected cache hit probability obtained via NA placement upper bounds the independent placement solution with probabilities $\mathbb{E}[z_{p_{k'} i}]$. NA has desirable properties in terms of sampling and concentration. Some important results that hold for independent variables, e.g., the Chernoff-Hoeffding bounds, and the Kolmogorov's inequality \cite{Dubhashi1996p2}, also hold for NA variables.

From Prop. \ref{NA}, it is clear that in terms of average cache hit performance, NA placement performs better than independent placement. Therefore, our main focus is on a class of placement policies that are NA. We also demonstrate that NA placement policies have lower variance across the nodes, hence are more stable than independent placement policies.

\begin{figure*}
    \centering
    \includegraphics[width=\textwidth]{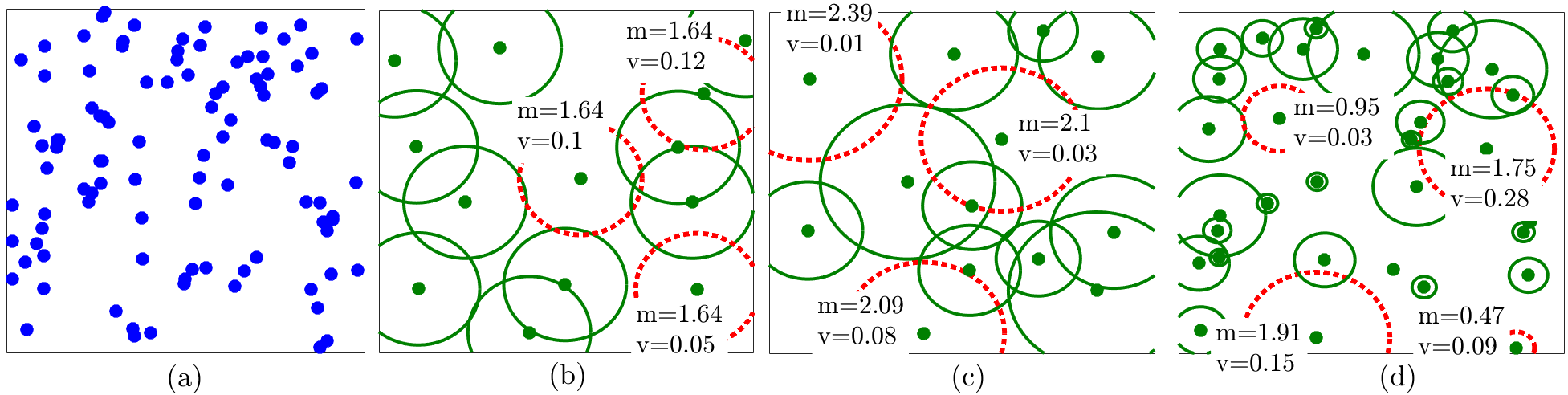}
    \caption{\small{SSCC p.p. realizations: (a) Begin with a realization of PPP $\Phi$. Associate a gamma distributed mark $m \sim \gamma(\alpha,\beta)$ to each point where $\bar{m}=\alpha\beta$ is the average mark and $\alpha\beta^2$ is its variance. As $\beta$ increases the mark variance increases. Associate a weight $v \sim  U[0, 1]$ to each node independently. A node $x\in\Phi$ is selected based on (\ref{retainingprobability}), where $p_0=1$, $f_c(r,m,n)$ as in (\ref{fcfunction}) with $c=100$. The marks and weights for some retained points are shown in dotted circles. The retained p.p.'s for (b) $\beta=0$ (fixed mark radii), (c) $\beta=0.1$ (mark radii have low variance), and (d) $\beta=1$ (mark radii have high variance). As $\beta$ increases, packing is denser, which is desired for spatially balanced caching.}}
    \label{fig:softcoremodel}
\end{figure*}

%%%
\section{A Soft-Core Caching Model}
The spatial soft-core caching (SSCC) policy is constructed from the underlying PPP $\Phi$ by removing certain nodes depending on the positions of the neighboring nodes, and on the marks and weights attached to them. It generalizes the Mat\'{e}rn II hard-core p.p. (MatII) such that there is a distinct distribution modeling the exclusion radius of each item.  

For each item $i$, let $\tilde{\Phi}_i=\{(x_k,m_k^{(i)},v_k^{(i)})\}_{k}$ be a homogeneous independently marked PPP with intensity $\lambda$, and i.i.d. $\mathbb{R}^2$-valued marks, where $\Phi=\{x_k\}$, and $\{(m_k^{(i)},v_k^{(i)})\}$ is the random bivariate mark. The first component $m^{(i)}$ of the bivariate mark is referred to as mark, and has distribution $\mu^{(i)}$. The mark of item $i$, i.e., $m^{(i)}$, denotes its exclusion radius, and depends on its popularity in the network. If item $i$ is more popular than item $j$, then $m^{(i)}$ is stochastically dominated\footnote{$X$ is stochastically dominated by $Y$, which is denoted by $X \leq^{st} Y$, if for all increasing functions $g$, we have $\mathbb{E}[g(X)]\leq \mathbb{E}[g(Y)]$.} by $m^{(j)}$. The second component $v^{(i)}$ of the bivariate mark is weight, which serves as a weight in the thinning procedure, and has distribution $\nu^{(i)}_{m^{(i)}}$ which might depend on $m^{(i)}$. 

Let $\Phi_{th,i}$ be a soft-core p.p. that denotes the set of points that cache item $i$. The cache placement model is such that item $i$ is stored in cache $x_k\in \Phi$ if and only if cache $x_k$ is kept as a point of $\Phi_{th,i}$. Equivalently, we have 
\begin{align}
z_{x_ki}=\mathbbm{1}\{i\in {\rm Cache}(x_k)\}=\mathbbm{1}\{x_k\in\Phi_{th,i}\}. 
\end{align}
Node $x_k$ is retained as a point of $\Phi_{th,i}$ with probability $\mathbb{E}[z_{x_ki}]=p(x_k,m_k^{(i)},v_k^{(i)},\Phi)$. The weights are i.i.d. and uniformly distributed, i.e. $v_k^{(i)}\sim U[0,1]$, for each node $x_k$ and item $i$. The marks $m_k^{(i)}$ are distributed according to $\mu^{(i)}$ for each $x_k$, and $i$. For the special case of MatII, i.e., when the marks are fixed, we optimized the exclusion radii in \cite{Malak2016twc}.

The number of items in cache $x_k$ is the sum of the individual items' indicator functions $C(x_k) = \sum\nolimits_i \mathbbm{1}\{i\in {\rm Cache}(x_k)\}$. The cache size constraint has to be satisfied on average, i.e.
\begin{align}
\label{avgcachesize}
N=\mathbb{E}[C(x_k)]=\sum\nolimits_i p(x_k,m_k^{(i)},v_k^{(i)},\Phi),\quad x_k\in\Phi.
\end{align}
We next detail the dependent thinning procedure, and investigate the relationship between $\Phi$ and $\Phi_{th,i}$, $i=\{1,\hdots, M\}$.

%%%
\subsection{Dependent Sampling of Nodes for Placement}
In this section and onwards, for brevity of notation, we omit the index $i$, and consider the generic thinned process $\Phi_{th}$, which is derived from $\Phi$ by applying the following probabilistic dependent thinning rule. Assume that mark $m$ has a distribution $\mu$, and $\vec{\delta}=\{m\}$ is the set of marks for all points in $\tilde{\Phi}$, where $m\sim \mu$ and $\bar{m}=\mathbb{E}_{m}[m]$. Assume that weight $\nu_m$ does not depend on the mark $m$. The marked point $(x, m, v) \in \tilde{\Phi}$ is retained as a point of $\Phi_{th}$ with probability
\begin{align}
\label{retainingprobability}
\!\!\!p(x, m, v, \Phi) = p_0\!\!\!\!\!\!\!\!\!\!\!\prod\limits_{(y,n,w)\in\Phi, y\neq x} \!\!\!\!\!\!\!\!\!\!\!\! [1 - \mathbbm{1}\{v \geq w\} f(||x - y||, m, n)]
\end{align}
independently from deleting or retaining other points of $\Phi$. In other words, a node $x \in \Phi$ is retained to cache item $i$ with probability $p_0$, if it has the lowest weight among all the points within its exclusion range. In (\ref{retainingprobability}), $p_0 \in (0, 1]$, $f : [0, \infty[ \times \mathbb{R}^2 \to [0, 1]$ is a deterministic function satisfying $f(\cdot, m, n) = f(\cdot, n, m)$ for all $m,\, n \in \mathbb{R}$. This means that if two points with marks $m$ and $n$, and weights $v\geq w$ are a distance $r > 0$ apart, then the point with weight $v$ is deleted by the other point with probability $f(r, m, n)$. Additionally, each surviving point is then again independently $p_0$-thinned. The function $f(||x - y||, m, n)$ in (\ref{retainingprobability}) should be determined according to (\ref{avgcachesize}). Inspired from  \cite{teichmann2013generalizations}, assume that $f$ satisfies 
\begin{comment}
\begin{align}
\label{fcfunction}
    f_c(r,m,n)=\begin{cases}
    1,\,\, 0\leq r\leq m+n,\\
    e^{-c(r-m-n)},\,\, r> m+n.
    \end{cases}
\end{align}
\end{comment}
\begin{align}
\label{fcfunction}
f_c(r,m,n)=\exp{(-c\lfloor r-m-n\rfloor_+)},\quad r\geq 0.
\end{align}

Denote by SSCC$[\lambda,\mu,(\nu_m)_{m\in\mathbb{R}}, p_0, f]$ the distribution of $\Phi_{th}$. We next give its intensity, i.e., $\lambda_{th}=\lambda \mathbb{E}[p(x, m, v, \Phi)]$.

\begin{theo}\cite[Theorem 12]{teichmann2013generalizations}
The intensity of the process $\Phi_{th}\sim$SSCC$[\lambda,\mu,(\nu_m)_{m\in\mathbb{R}}, p_0, f]$ is given by 
\begin{multline}
\lambda_{th}=\lambda p_0 \int\nolimits_{\mathbb{R}} \int\nolimits_{\mathbb{R}} \exp\Big(-\lambda \int\nolimits_{\mathbb{R}} F_{\nu_n}(w)  \nonumber\\ \int\nolimits_{\mathbb{R}^2}  f(||x ||, m, n){\rm d}x\ \mu({\rm d} n ) \Big)\, \nu_m({\rm d}w)\,\mu({\rm d}m),
\end{multline}
where $F_{\nu_m}$ is the cumulative distribution function of $\nu_m$. 
\end{theo}

\begin{proof}
The probability generating functional (PGFL) \cite{Stoyan1996} of the PPP states for function $f(x)$ that $\mathbb{E}\left[\prod\nolimits_{x\in\Phi} f(x)\right] = \exp\big(-\lambda \int\nolimits_{\mathbb{R}^2} (1 - f(x)){\rm d}x\big)$. We obtain $\lambda_{th}$ using the PGFL and  $\mathbb{E}[\mathbbm{1}\{\nu_n \leq w\}]=\int\nolimits_{\mathbb{R}} F_{\nu_n}(w) \mu({\rm d}n)$, along with (\ref{retainingprobability}). \begin{comment}we compute the intensity of $\Phi_{th}$ as
\begin{align}
\lambda_{th}
=\!\lambda p_0 
\int\nolimits_{\mathbb{R}} \!\int\nolimits_{\mathbb{R}} e^{-\lambda \mathbb{E}\big[ \mathbbm{1}\{\nu_n\leq w\}  \int\nolimits_{\mathbb{R}^2}  f(||x ||, m, n){\rm d}x\ \big] }\, \nu_m({\rm d}w)\,\mu({\rm d}m).\nonumber
\end{align}
\end{comment}
\end{proof}
In Fig. \ref{fig:softcoremodel}, we plot different realizations of SSCC $\Phi_{th}$ formed by thinning $\Phi$. As the mark variance increases, the packing is denser, which is desired for spatially balanced caching.

%%%
\subsection{Spherical Contact Distribution Function}
Our goal in this section is to relate the cache hit probability distribution to the (spherical) contact distribution function. 

\begin{defi}
The spherical contact distribution function (SCDF) of the p.p. $\Xi$ is the conditional distribution function of the distance from a point chosen randomly outside $\Xi$ (i.e. $0$), to the nearest point of $\Xi$ given $0\notin \Xi$ \cite{Stoyan1996}. It is given by 
\begin{align}
\label{sphericalCDF}
H(r) =\mathbb{P}(\Rsp\leq r\vert \Rsp>0), \quad r \geq 0,
\end{align}
where $\Rsp=\inf\{s: \Xi\cap sA\neq \emptyset\}$, where $A=B_0(1)$, and $rA$ is the dilation of the set $A$ by the factor $r$. 
\end{defi}

As an example, in Fig. \ref{fig:CDF} we show the SCDF for the Boolean model with random spherical grains in \cite[Ch. 3.1]{BaccelliBook1}.

\begin{theo}\label{CachehitvsSCDF} The average cache hit probability of policy $\pi$ is
\begin{align}
\mathbb{E}_{\pi}[F(Z)]=\mathbb{E}_{\mathcal{I}}[\HpiI(\Rdd)],
\end{align}
where $\HpiI(\Rdd)$ is the SCDF of the thinned p.p. $\Phi_{th,\mathcal{I}}$ for $\mathcal{I}$.
\end{theo}

\begin{proof}
Let $B=B_0(\Rdd)$ and $\Phi_{th,i}(B)=\sum\nolimits_{x\in\Phi_{th,i}}1(x\in B)$ be the number of transmitters containing item $i$ within a circular region of radius $\Rdd$ around the origin. Then we have
\begin{align}
F(Z)=\sum\nolimits_{i}{p_r(i)\mathbbm{1}(\Phi_{th,i}(B)>0)}.\nonumber
\end{align}
The average cache hit probability is given by $\mathbb{E}[F(Z)]=\sum\nolimits_{i}{p_r(i)\mathbb{P}(\Phi_{th,i}(B)>0)}$, where defining $\Rsp=\inf\{s: \Phi_{th,i}(B_0(s)) \neq 0\}$, given $0\notin \Phi_{th,i}$ we have that
\begin{align}
\label{hit_cdf_equivalence}
    \mathbb{P}(\Phi_{th,i}(B)>0) = \mathbb{P}(\Rsp\leq \Rdd \vert \Rsp>0),
\end{align}
which is the SCDF of $\Phi_{th,i}$ evaluated at $\Rdd$. 
\end{proof}

The variance of $F(Z)$ across the nodes satisfies
\begin{align}
\Varpi{F(Z)}=\sum\nolimits_{i} p^2_r(i) \Hpii(\Rdd) (1-\Hpii(\Rdd)) 
\end{align}
since the spatial thinning processes across different items are independent of each other. Under the IRM and a Zipf popularity model, $\Varpi{F(Z)}$ decreases with increasing variance of marks when $\mathbb{E}[C(x)]$ is held constant. A spatially balanced sampling yields a low $\Varpi{F(Z)}$ as expected.

%%%
\subsection{Migration to the Child Process: Effective Thinning}
Consider the pair $\Phi-\Phi_{th}$ of mother and child p.p.'s. The spherical contact distance denotes the distance between a typical point in $\Phi$ and its nearest neighbor from $\Phi_{th}$.

Using (\ref{sphericalCDF}), the SCDF for the p.p. $\Phi$ can be written as:
\begin{align}
\label{HRintegral}
\Hpi(R)=1 - \exp\Big( -\int\nolimits_{0}^{R}2\pi r\lambda \etapi(r,\delta)  {\rm d}r \Big), 
\end{align}
where $\etapi(r,\delta)$ is the conditional thinning Palm-probability (CTPP), i.e. the probability of the point $x\in\Phi$ migrating to $\Phi_{th}$ under policy $\pi$, with a fixed (exclusion) radius $\delta$. It equals
\begin{align}
\label{ConditionalThinningPalm-Probability}
    \etapi(r,\delta)=\mathbb{P}(x\in \Phi_{th} \vert \Phi_{th}\cap B_{x_0}(r)=\emptyset, x_0\in \Phi).
\end{align}

\begin{remark}
An effective thinning policy yields a larger CTPP $\etapi(r,\delta)$. The more effective the thinning is, the larger (\ref{HRintegral}) is. From Theorem \ref{CachehitvsSCDF}, $\mathbb{E}_{\pi}[F(Z)]$ is improved if $\pi$ is more effective.
\end{remark}

We next compute the CTPP for the SSCC policy.

%%%
\begin{prop} \label{CTPP_SSCC}
The CTPP for PPP-SSCC is given as 
\begin{align}
    \etaGM(r,\vec{\delta})=
    \int\nolimits_{\mathbb{R}}\int\nolimits_{0}^1
    e^{-u \lambda \int\nolimits_{\mathbb{R}} \int\nolimits_{\mathbb{R}^2} h(||x||,m,n)
    {\rm d}x \, \mu({\rm d}n) }\, {\rm d}u\, \mu({\rm d}m),\nonumber
\end{align}
where given radius marks $m,\,n$, $h(||x||,m,n)$ satisfies the relation $\int\nolimits_{\mathbb{R}^2} h(||x||,m,n)\,{\rm d}x=\pi(m+n)^2-l_2(r,n)$, where $l_2(r,\delta)$ is the area of the intersection of $B_{x_0}(r)$ and $B_{x}(\delta)$. 
\begin{comment}
It is given by
\begin{align}
    l_2(r,\delta)=\begin{cases}
    \pi r^2,\quad 0<r<\frac{\delta}{2}\\
    r^2\cos^{-1}\big(1-\frac{\delta^2}{2r^2}\big)+\delta^2\cos^{-1}\big(\frac{\delta}{2r}\big)\nonumber\\
    -\frac{\delta}{2}\sqrt{4r^2-\delta^2},\quad r\geq \delta/2.
    \end{cases}
\end{align}
\end{comment} 
\end{prop}
\begin{proof}
The proof follows from generalizing \cite[Eq. (15)]{al2016nearest}.
\end{proof}

%%%
\begin{cor}
The CTPP for PPP-MatII is given as 
\begin{align}
    \etaM(r,\delta)
    =\frac{1-e^{-\lambda (\pi\delta^2-l_2(r,\delta))}}{\lambda (\pi\delta^2-l_2(r,\delta))}.\nonumber
\end{align}
\end{cor}
%%%%

%%%
The next Theorem shows that having a distribution on the marks yields a more effective thinning than MatII does.
\begin{theo}\label{convexity}
The CTPPs satisfy $\etaGM(r,\vec{\delta})\geq \etaM(r,\bar{m})$, where $\vec{\delta}=\{m\}$ is the set of marks in $\tilde{\Phi}$, with $\bar{m}=\mathbb{E}_{m}[m]$, . 
\end{theo}

\begin{proof}
From Prop. \ref{CTPP_SSCC}, we have that
\begin{align}
\etaGM(r,\vec{\delta})=\int\nolimits_{\mathbb{R}}\int\nolimits_{0}^1
    e^{-u \lambda \int\nolimits_{\mathbb{R}} (\pi (m+n)^2-l_2(r,n)) \, 
    \mu({\rm d}n) }\, {\rm d}u\, \mu({\rm d}m) \nonumber\\
    =\mathbb{E}_{m}\Big[\mathbb{E}_{U}\Big[
    e^{-U q(\lambda,r,m)}\Big]\Big] 
    =\mathbb{E}_{m}\left[\frac{1-e^{-q(\lambda,r,m)}}{q(\lambda,r,m)}\right],\nonumber
\end{align}
where $U\sim U[0,1]$, and $q(\lambda,r,m)=\lambda\pi \big(m^2+2m\bar{m}_2\big)+\lambda \mathbb{E}_{m_2}\big[\pi m_2^2-l_2(r,m_2)\big]$. Let $f=e^{-x}$, $x=U\lambda\pi(m^2+2m\bar{m}_2)$. We have, $\etaGM(r,\vec{\delta})=\mathbb{E}_{m}[g(m)]$, with $g=\frac{f-1}{log(f)}=\frac{1-e^{-x}}{x}$. Then $g'=\frac{e^{-x}(x+1)-1}{x^2}$, $g''=\frac{2-e^{-x}[x^2+2x+2]}{x^3}>0$. Hence, $\etaGM(r,\vec{\delta})=\mathbb{E}_{m}[g(m)]\geq g(\bar{m})=\etaM(r,\bar{m})$.
\end{proof}

\begin{figure}[t!]
    \centering
    \includegraphics[width=0.3\textwidth]{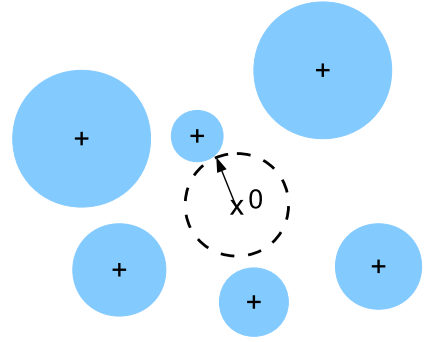}
    \caption{\small{The radius of the smallest sphere centered at $0$ and intersecting the Boolean Model $\Phi_{th}$. The SCDF is the conditional distribution function of the radius of the sphere, given $0\notin\Phi_{th}$ \cite[Ch. 3.1]{BaccelliBook1}.}}
    \label{fig:CDF}
\end{figure}

Exploiting Theorem \ref{convexity}, $\etaGM(r,\vec{\delta})$ can be improved using a mixture of marks. The variable exclusion range model can suit to the case of cellular networks where demand is not uniform across the network \cite{zhou2015spatial}, which we left as future work.

%%%%%%
\subsection{Cache Over-Utilization}\label{overutilization}
The cache placement requires $C(x)=\sum\nolimits_{i}\mathbbm{1}_{x\in\Phi_{th,i}}\leq N$, for all $x\in\Phi$, where $N$ is finite. The storage constraint is satisfied on average, i.e. $N=\mathbb{E}[C(x)]=\sum_i p(x,m_i,v,\Phi)$, $x\in\Phi$. However, the set of child p.p.'s $\{\Phi_{th,i}\}_i$, $i=1,\hdots, M$ might overlap. We need to make sure that the cache capacities are not over-utilized. Hence, the intersection of the sampled processes, i.e. $\cap_i \Phi_{th,i}$, should not include any $x\in\Phi$ more than $N$ times with high probability. We next provide an upper bound for the violation probability of the cache size for SSCC. 

\begin{figure*}[t!]
    \centering
    \includegraphics[width=0.49\textwidth]{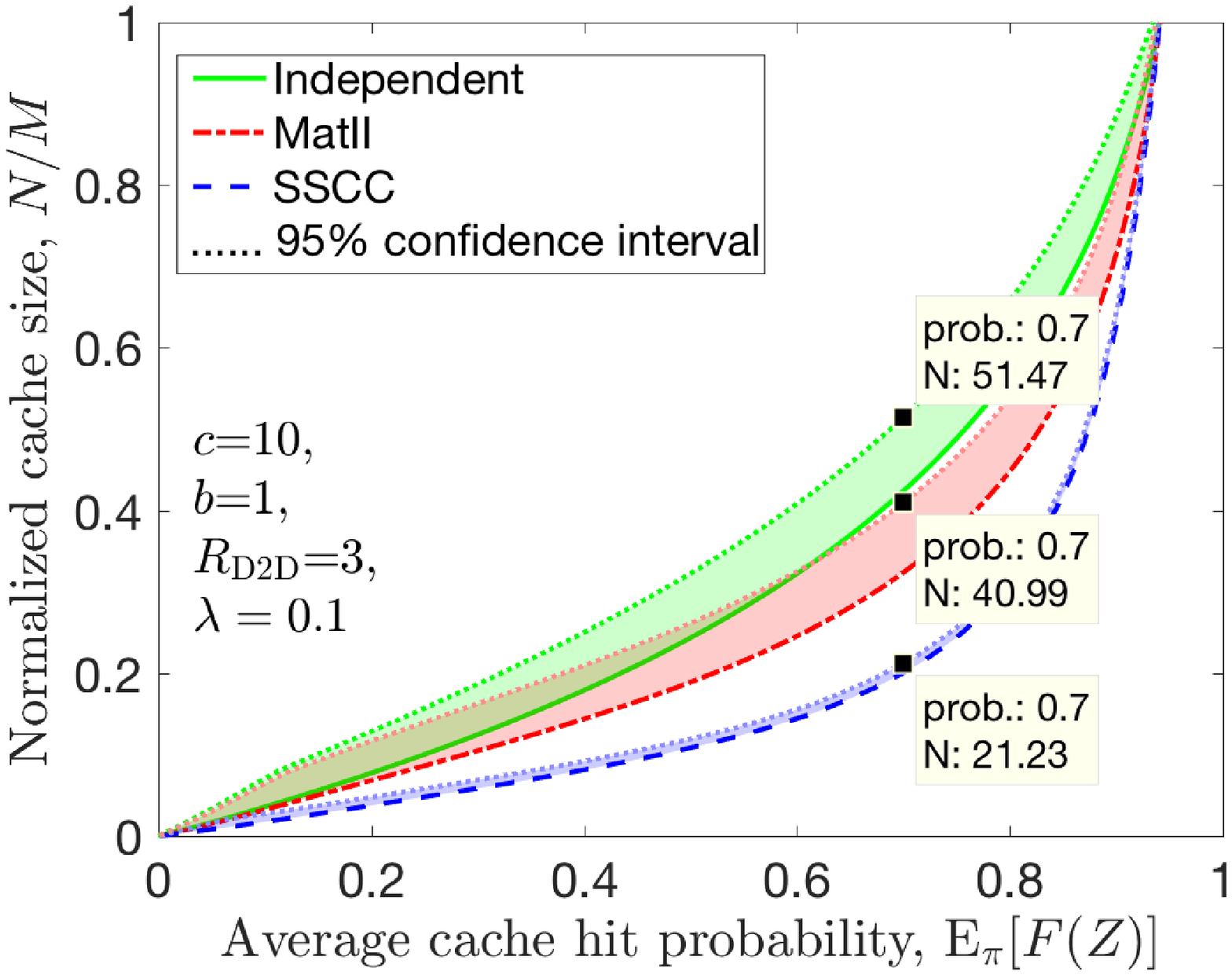}
    \includegraphics[width=0.49\textwidth]{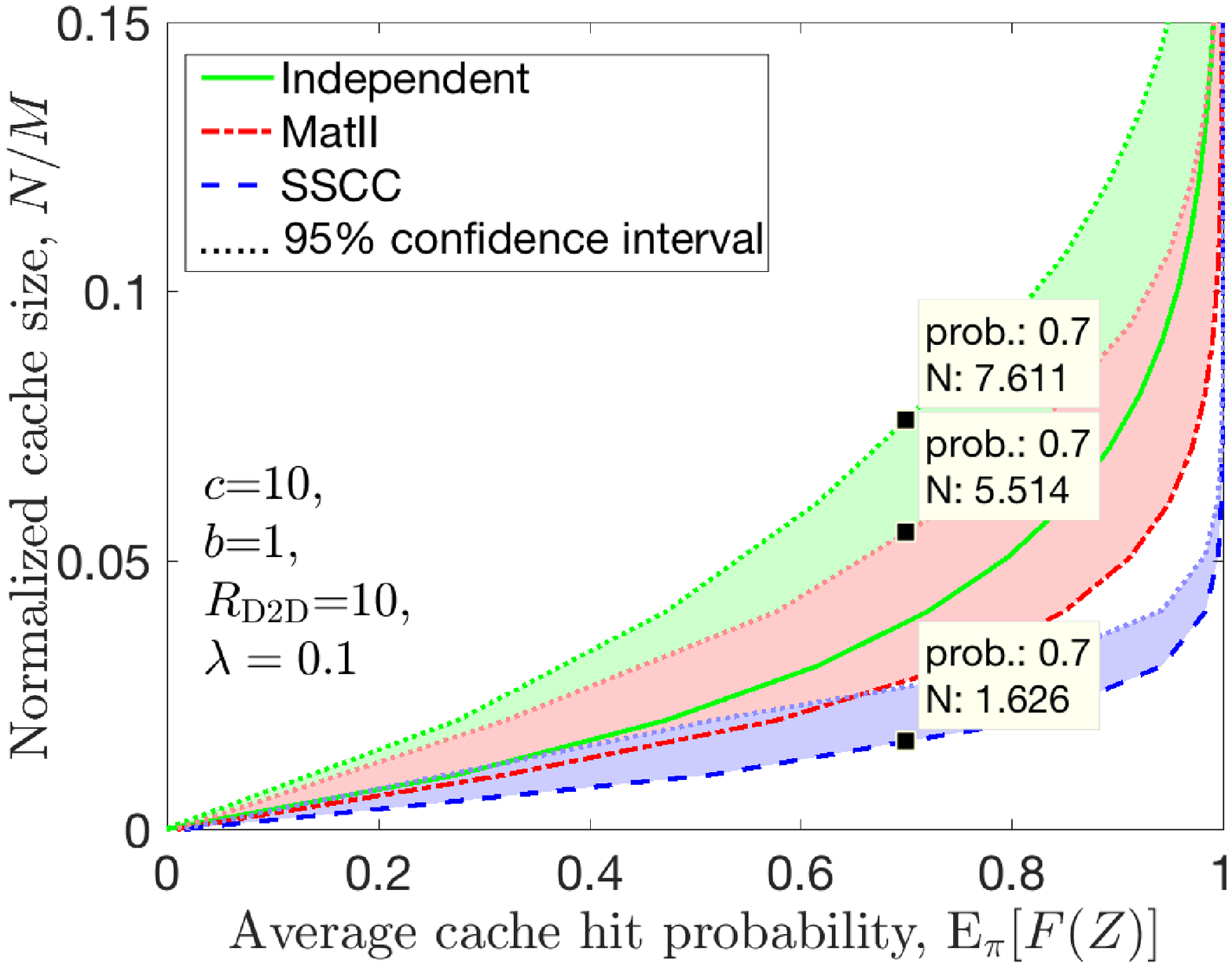}
    \caption{\small{The normalized cache size versus average cache hit rate for different placement policies. (Left) $\Rdd=3$, (Right) $\Rdd=10$.}}
    \label{fig:performance1}
\end{figure*}

\begin{comment}
\begin{prop}{\bf Chernoff bound for cache size.} 
The cache violation probability is upper bounded as
\begin{align}
    \mathbb{P}(C(x)>C)\leq \exp{\Big(\epsilon - (N+\epsilon)\log\Big(1+\frac{\epsilon}{N }\Big)\Big)},
\end{align}
where $C=N+\epsilon$ for $\epsilon$ arbitrarily small.
\end{prop}

\begin{proof} We can rewrite the cache violation probability as
\begin{align}
    &\mathbb{P}(C(x)>C)
    =\mathbb{P}\Big(\sum\limits_{i}\mathbbm{1}_{x\in\Phi_{th,i}}>C\Big)\nonumber\\
    &=\sum\limits_{A\subseteq S\,: |A|>C}\mathbb{P}(x\in\Phi_{th,i},\,i\in A,\, x\notin\Phi_{th,j},\,j\in A^{\mathsf{c}})\nonumber\\
    &\overset{(a)}{=}\sum\limits_{A\subseteq S\,: |A|>C}\,\,\prod\limits_{i\in A}\mathbb{P}(x\in\Phi_{th,i})\prod\limits_{j\in A^{\mathsf{c}}}(1-\mathbb{P}(x\in\Phi_{th,j}))\nonumber\\
    &\overset{(b)}{\leq} \exp{\Big(C-\sum\limits_i p(x,m_i,v,\Phi) - C\log\Big(\frac{C}{\sum\limits_i p(x,m_i,v,\Phi) }\Big)\Big)},\nonumber
\end{align}
where $C=N+\epsilon$ for $\epsilon$ arbitrarily small, $S$ is the set of all subsets of $\{1,\hdots, M\}$, and $A$ is a subset of $S$ and $A^{\mathsf{c}}$ is its complement, and $(a)$ is due to that $\mathbb{P}(x\in\Phi_{th,i},\,x\in\Phi_{th,j})=\mathbb{P}(x\in\Phi_{th,i})\mathbb{P}(x\in\Phi_{th,j})$ for all $i\neq j$ due to independent sampling of points for each $i$, and $(b)$ follows from employing the Chernoff bound.
\end{proof}

As the Chernoff bound does not capture the second-order characteristics of the sampling, we next provide another bound based on the Bernstein inequalities.
\end{comment}

\begin{prop}{\bf Bernstein bound for cache size.} 
The cache violation probability is upper bounded as
\begin{align}
\label{Bernstein}
\mathbb{P}(C(x)>C) \leq \exp\left( -\frac{(C-N)^2}{{\rm Var}[C(x)]+\frac{1}{3}(C-N)} \right),
\end{align}
where ${\rm Var}[C(x)]=\sum\nolimits_{i=1}^M{\rm Var}[z_{xi}]$ since the placement is independent across items, where ${\rm Var}[z_{xi}]=\mathbb{E}[z_{xi}^2]-\mathbb{E}[z_{xi}]^2=p(x,m_i,v,\Phi)(1-p(x,m_i,v,\Phi))$, for $i\in\{1,\hdots, M\}$, $x\in\Phi$.
\end{prop}
\begin{proof}
It follows from employing Bernstein inequality since $\mathbbm{1}_{x\in\Phi_{th,i}}$ are independent $1$-$0$ random variables across $i$.
\end{proof}

As ${\rm Var}[C(x)]$ drops, the bound in (\ref{Bernstein}) becomes lower. Hence, the cache violation probability is negligible if the cache placement strategy has very low-variance. In Sect. \ref{experiments}, we demonstrate that SSCC has very small violation probability.

For the spatially independent placement policy in \cite{Blaszczyszyn2014}, where nodes are sampled i.i.d., authors have proposed a probabilistic placement technique to guarantee that the cache constraint is satisfied with equality. However, in SSCC, nodes are not sampled independently. Because the placement policy is NA across the nodes, it is nontrivial to design probabilistic placement techniques to satisfy the cache size constraint. In this section, we discuss how to bound the violation probability, and demonstrate in Sect. \ref{experiments} that for SSCC the cache violation probability can be made negligibly small.

%%%
\section{Numerical Simulations}
\label{experiments}
The nodes live in a square region of the Euclidean plane with area $L^2$ where $L=100$. To avoid edge effects, we evaluate the performance only for the middle square region with area $L^2/9$. The network parameters are $\lambda=0.1$ and $\Rdd\in\{3,\,10\}$. The request process is isotropic and Zipf distributed with parameter $\gamma_r=0.1$ over $M=100$ items. 

For MatII, there is a fixed exclusion range for a given item, and we have derived the optimal exclusion radii in \cite{Malak2016twc}. Let $r_i$ be the optimized exclusion range for item $i$ for MatII. For SSCC, we assume that the marks $m^{(i)}$ for item $i$ (exclusion radii) are distributed according to a gamma distribution $\mu^{(i)}=\Gamma(0.7 r_i,1)$ for each $x\in\Phi$, and all items $i$, where we choose its parameters such that the average value of the radius mark for item $i$ equals $\bar{m}^{(i)}=0.7 r_i$. Hence, $\Phi_{th}\sim$SSCC$[0.1,\Gamma(0.7 r_i,1),U[0,1],1,f_{10}]$. We can observe that the SSCC model can be used to optimize the cache hit probability-cache violation probability tradeoff. As variance of exclusion range increases, the violation probability might also increase for a desired cache hit probability. Note that we do not optimize the distributions of the marks $\mu^{(i)}$ across all $i$ over a class of distributions. We leave the study of the fundamental performance limits of SSCC as future work.   

We numerically investigate how much cache over-provisioning is required for different spatial cache placement policies: spatially independent \cite{Blaszczyszyn2014}, MatII \cite{Malak2016twc}, and SSCC cache placement. In Fig. \ref{fig:performance1}, we investigate the required cache size $N$ (normalized) of each policy given that the probability of cache violation is small such that $\mathbb{P}[|C(x)-N|\leq \epsilon] > 0.95$ in order to characterize the required cache size for a given average cache hit probability. We also illustrate the $95\%$ confidence intervals represented by the shaded regions, and mark the cache sizes for different policies when the average cache hit probability is $\mathbb{E}_{\pi}[F(Z)]=0.7$. For example, when $\Rdd=3$, for the $95\%$ confidence interval, the excess cache ratio for independent placement in \cite{Blaszczyszyn2014}, and MatII placement in \cite{Malak2016twc} with respect to the SSCC policy is $142\%$, and $93\%$, respectively. When we have $\Rdd=10$, the respective excess ratios for the independent and MatII placement policies are $188\%$, and $109\%$, which are illustrated on the plots. SSCC yields a better concentration of the required cache size, which is desired. Hence, policies like SSCC can be exploited so that the cache does not overrun or underrun its capacity constraint. 

SSCC gives insights into not only how to cache the content, but also how to effectively sample in spatial settings. SSCC is suited for enabling applications such as D2D and P2P as it promotes the item diversity and reciprocation. Extensions include the incorporation of the spatial variation of the demand. They also include employing the exclusion based models to optimize the performance of time-to-live (TTL) caches.

\section*{Acknowledgment}
We thank Salman Salamatian for helpful discussions.

\bibliographystyle{IEEEtran}
\bibliography{references}

\begin{thebibliography}{10}
\providecommand{\url}[1]{#1}
\csname url@samestyle\endcsname
\providecommand{\newblock}{\relax}
\providecommand{\bibinfo}[2]{#2}
\providecommand{\BIBentrySTDinterwordspacing}{\spaceskip=0pt\relax}
\providecommand{\BIBentryALTinterwordstretchfactor}{4}
\providecommand{\BIBentryALTinterwordspacing}{\spaceskip=\fontdimen2\font plus
\BIBentryALTinterwordstretchfactor\fontdimen3\font minus
  \fontdimen4\font\relax}
\providecommand{\BIBforeignlanguage}[2]{{%
\expandafter\ifx\csname l@#1\endcsname\relax
\typeout{** WARNING: IEEEtran.bst: No hyphenation pattern has been}%
\typeout{** loaded for the language `#1'. Using the pattern for}%
\typeout{** the default language instead.}%
\else
\language=\csname l@#1\endcsname
\fi
#2}}
\providecommand{\BIBdecl}{\relax}
\BIBdecl

\bibitem{Shanmugam2013}
N.~G. K.~Shanmugam, A.~G. Dimakis, A.~F. Molisch, and G.~Caire,
  ``{FemtoCaching}: {Wireless} content delivery through distributed caching
  helpers,'' \emph{IEEE Trans. Inf. Theory}, vol.~59, no.~12, pp. 8402--13,
  Dec. 2013.

\bibitem{MaddahAli2013Journal}
M.~A. Maddah-Ali and U.~Niesen, ``Fundamental limits of caching,'' \emph{IEEE
  Trans. Inf. Theory}, vol.~60, no.~5, pp. 2856--67, May 2014.

\bibitem{yu2018exact}
Q.~Yu, M.~A. Maddah-Ali, and A.~S. Avestimehr, ``The exact rate-memory tradeoff
  for caching with uncoded prefetching,'' \emph{IEEE Trans. Inf. Theory},
  vol.~64, no.~2, pp. 1281--1296, 2018.

\bibitem{Ji2014}
M.~Ji, G.~Caire, and A.~F. Molisch, ``Fundamental limits of caching in wireless
  {D2D} networks,'' \emph{IEEE Trans. Inf. Theory}, vol.~62, no.~2, pp.
  849--869, Feb. 2016.

\bibitem{maddah2015cache}
M.~A. Maddah-Ali and U.~Niesen, ``Cache-aided interference channels,'' in
  \emph{Proc., IEEE Int. Sym. Inf. Theory}.\hskip 1em plus 0.5em minus
  0.4em\relax IEEE, 2015, pp. 809--813.

\bibitem{park2016joint}
S.-H. Park, O.~Simeone, and S.~Shamai, ``Joint optimization of cloud and edge
  processing for fog radio access networks,'' in \emph{Proc., IEEE Int. Sym.
  Inf. Theory}, 2016, pp. 315--319.

\bibitem{Che2002}
H.~Che, Y.~Tung, and Z.~Wang, ``Hierarchical web caching systems: {Modeling},
  design and experimental results,'' \emph{IEEE J. Sel. Areas Commun.},
  vol.~20, no.~7, pp. 1305--14, Sep. 2002.

\bibitem{Giovanidis2016}
A.~Giovanidis and A.~Avranas, ``Spatial multi-{LRU} caching for wireless
  networks with coverage overlaps,'' in \emph{Proc., ACM Sigmetrics/IFIP
  Performance}, Antibes, France, Jun. 2016, pp. 403--405.

\bibitem{BerHenCiuSch2015}
D.~S. Berger, S.~Henningsen, F.~Ciucu, and J.~B. Schmitt, ``Maximizing cache
  hit ratios by variance reduction,'' in \emph{ACM Sigmetrics Performance
  Evaluation Review}, vol.~43, no.~2, Sep. 2015, pp. 57--59.

\bibitem{LegKelMic2006}
A.~Legout, G.~Urvoy-Keller, and P.~Michiardi, ``Rarest first and choke
  algorithms are enough,'' in \emph{Proc., ACM Sigcomm}, Oct. 2006, pp.
  203--216.

\bibitem{Blaszczyszyn2014}
B.~B\l{}aszczyszyn and A.~Giovanidis, ``Optimal geographic caching in cellular
  networks,'' in \emph{Proc., IEEE ICC}, UK, Jun. 2015, pp. 3358--3363.

\bibitem{IoannidisYeh2016}
S.~Ioannidis and E.~M. Yeh, ``Adaptive caching networks with optimality
  guarantees,'' in \emph{Proc., ACM SIGMETRICS}, Jun. 2016, pp. 113 -- 124.

\bibitem{Malak2016twc}
D.~Malak, M.~Al-Shalash, and J.~G. Andrews, ``Spatially correlated content
  caching for device-to-device communications,'' \emph{IEEE Trans. Wireless
  Commun.}, vol.~17, no.~1, pp. 56--70, Jan. 2018.

\bibitem{zhou2015spatial}
S.~Zhou, D.~Lee, B.~Leng, X.~Zhou, H.~Zhang, and Z.~Niu, ``On the spatial
  distribution of base stations and its relation to the traffic density in
  cellular networks.'' \emph{IEEE Access}, vol.~3, pp. 998--1010, 2015.

\bibitem{traverso2013temporal}
S.~Traverso, M.~Ahmed, M.~Garetto, P.~Giaccone, E.~Leonardi, and S.~Niccolini,
  ``Temporal locality in today's content caching: why it matters and how to
  model it,'' \emph{ACM Sigcomm Computer Commun. Review}, vol.~43, no.~5, pp.
  5--12, 2013.

\bibitem{Dubhashi1996p2}
D.~P. Dubhashi and D.~Ranjan, ``{Balls and bins: A study in negative
  dependence},'' \emph{BRICS Report Series}, vol.~3, 1996.

\bibitem{teichmann2013generalizations}
J.~Teichmann, F.~Ballani, and K.~G. van~den Boogaart, ``Generalizations of
  {Mat{\'e}rn's} hard-core point processes,'' \emph{Spatial Statistics},
  vol.~3, pp. 33--53, 2013.

\bibitem{Stoyan1996}
D.~Stoyan, W.~Kendall, and J.~Mecke, \emph{Stochastic Geometry and Its
  Applications}, 2nd~ed.\hskip 1em plus 0.5em minus 0.4em\relax John Wiley and
  Sons, 1996.

\bibitem{BaccelliBook1}
F.~Baccelli and B.~B{\l}aszczyszyn, \emph{{Stochastic Geometry and Wireless
  Networks}}.\hskip 1em plus 0.5em minus 0.4em\relax NOW: Found. Trends.
  Network., 2010.

\bibitem{al2016nearest}
A.~Al-Hourani, R.~J. Evans, and S.~Kandeepan, ``Nearest neighbor distance
  distribution in hard-core point processes,'' \emph{IEEE Communications
  Letters}, vol.~20, no.~9, pp. 1872--1875, 2016.

\end{thebibliography}

\end{document}